\documentclass[11pt]{article}

\usepackage{hyperref}	
\usepackage{authblk}
\usepackage[margin=1.1in]{geometry}
\usepackage[nodayofweek]{datetime}
\usepackage{amsfonts}
\usepackage{amsmath,amssymb}
\usepackage[amsmath,thmmarks,amsthm]{ntheorem}
\usepackage{mathtools}
\usepackage{enumitem} 
\usepackage{pgfplots}
\usepackage{cleveref}
\usepackage{tikz}
\usetikzlibrary{calc,shapes}

\setenumerate[1]{label={(\alph{enumi})}} 

\newcommand{\ALG}{\textit{ALG}}
\newcommand{\cost}{\textit{cost}}
\newcommand{\ie}{i.\,e.\,}
\newcommand{\OPT}{\textit{OPT}}

\newtheorem{theorem}{Theorem}
\newtheorem{proposition}[theorem]{Proposition}
\newtheorem{corollary}[theorem]{Corollary}
\newtheorem{lemma}[theorem]{Lemma}
\newtheorem{definition}{Definition}

\DeclarePairedDelimiter{\floor}{\lfloor}{\rfloor}
	
\title{The Infinite Server Problem\thanks{Supported by the ERC Advanced Grant 321171 (ALGAME) and by EPSRC.}}

\author[1]{Christian Coester}
\author[1]{Elias Koutsoupias}
\author[1]{Philip Lazos}
\affil[1]{Department of Computer Science, University of Oxford}

\begin{document}

\usdate
\maketitle


\begin{abstract}
	We study a variant of the $k$-server problem, the infinite server problem,
	in which infinitely many servers reside initially at a particular
	point of the metric space and serve a sequence of requests. In the
	framework of competitive analysis, we show a surprisingly tight
	connection between this problem and the $(h,k)$-server problem, in
	which an online algorithm with $k$ servers competes against an offline
	algorithm with $h$ servers. Specifically, we show that the infinite
	server problem has bounded competitive ratio if and only if the
	$(h,k)$-server problem has bounded competitive ratio for some
	$k=O(h)$. We give a lower bound of $3.146$ for the competitive ratio
	of the infinite server problem, which implies the same lower bound for
	the $(h,k)$-server problem even when $k/h \to \infty$ and holds also for the line metric; the previous
	known bounds were 2.4 for general metric spaces and 2 for the line. For weighted trees and layered graphs we obtain upper bounds, although they depend on the depth. Of particular interest is the infinite server problem on the line, which we show to be equivalent to the seemingly easier case in which all requests are in a fixed bounded interval away
	from the original position of the servers. This is a special case of a more general reduction from arbitrary metric spaces to bounded subspaces. Unfortunately, classical approaches (double coverage and generalizations, work function algorithm, balancing algorithms) fail even for this special case.
\end{abstract}

\section{Introduction}
The $k$-server problem is a fundamental well-studied
online problem \cite{ManasseMS88,Koutsoupias09}.
In this problem $k$ servers serve a
sequence of requests. The servers reside at $k$ points of a metric
space $M$ and requests are simply points of $M$. Serving a request
entails moving one of the servers to the request. The objective is to
minimize the total distance traveled by the servers. The most
interesting variant of the problem is its online version, in which the
requests appear one-by-one and the online algorithm must decide how to
serve a request without knowing the future requests. It is known
that the deterministic $k$-server problem has competitive ratio
between $k$ and $2k-1$ for every metric space with at least $k+1$
distinct points \cite{ManasseMS88,KoutsoupiasP95}.

In this paper, we study the \emph{infinite server problem}, the
variant of the $k$-server problem in which there are infinitely many
servers, all of them initially residing at a given point, the
\emph{source}\footnote{We first learned about this problem from Kamal
	Jain \cite{jain}.}. At first glance it
may appear that the lower bound of $k$ for the $k$-server problem
would imply an unbounded competitive ratio for the infinite server
problem. But consider, for example, the version of
the $k$-server problem on uniform metric spaces (\ie the distance between any two points is $1$), and observe that the
infinite server problem has competitive ratio 1 for this case.

The infinite server problem is closely related to the $(h,k)$-server problem,
the resource augmentation version of the $k$-server problem in which
the online algorithm has $k$ servers and competes against an offline
algorithm for $h\leq k$ servers. This model is also known as \emph{weak
	adversaries} \cite{BansalEJKP15,Koutsoupias99}. One major open question 
in competitive
analysis is whether the $(h,k)$-server problem has bounded competitive
ratio when $k\gg h$. Bar-Noy and Schieber (see \cite[p.\,175]{BorodinE98}) showed that when
$h=2$, the competitive ratio on the line metric is $2$ for \emph{any} $k$, and recently,
Bansal et al.~\cite{BansalEJK17} showed a lower bound of 2.4 for the general case
$k\gg h$. Here we show a, perhaps surprising, tight connection between
the infinite server problem and the $(h,k)$-server problem, which
allows us to improve both lower bounds to 3.146.

The infinite server problem is also a considerable generalization of
the ski-rental problem, since the ski-rental problem is essentially a
special case of the infinite server problem when the metric space is
an isosceles triangle.

\subsection{Previous Work}
The $k$-server problem was first formulated by Manasse et 
al.~\cite{ManasseMS88}, to generalize a variety of online setting whose 
stepwise cost had a `metric'-like structure. They build on previous work by 
Sleator and Tarjan \cite{SleatorT85}, the genesis of competitive analysis, on 
the paging problem. This problem can be easily recast as a $k$-server 
instance for the uniform metric and was already known to be $k$-competitive. 

Manasse et al.~\cite{ManasseMS88} also showed that  the competitive ratio 
of the $k$-server problem is at least 
$k$ on any metric space with more than $k$ points. 
They then proposed the renowned $k$-server conjecture, stating that this 
bound is tight. This 
has been shown to be true for $k=2$ \cite{ManasseMS88} and for several 
special metric spaces 
\cite{ChrobakKPV91,ChrobakL91,KoutsoupiasP96,ManasseMS88,SleatorT85}.
A stream of refinements \cite{fiat1990competitive,bartal2000harmonic} 
lead to 
better 
competitive ratios for general metric 
spaces until \cite{KoutsoupiasP95}
showed that a competitive ratio of $2k-1$ 
can be achieved on any metric space. 
Chasing the competitive ratio for the deterministic (and randomized) 
$k$-server problem has been pivotal for the development of
competitive analysis.
For a more in depth view on the history of the $k$-server problem and 
further related work, we refer to \cite{Koutsoupias09}.

In the weak adversaries setting, significantly less is known. For the 
$(h,k)$-server problem, the exact competitive ratio is 
$\frac{k}{k-h+1}$ on uniform metrics (equivalent to the paging problem) 
\cite{SleatorT85} and weighted star metrics (equivalent to weighted paging) 
\cite{Young94}. Bansal et al. \cite{BansalEJK17} showed recently for weighted trees that the competitive ratio as 
$k/h\to\infty$ can be bounded by a constant depending on the depth of 
the tree. On general metrics, the $(h,k)$-server 
problem is still very poorly understood. No algorithm is known for general 
metrics that performs better than disabling the $k-h$ extra servers and 
using $h$ servers only. In fact, for the line it was shown 
\cite{BansalEJKP15,BansalEJK17} that the Double Coverage 
Algorithm and the Work Function Algorithm -- despite achieving the 
optimal competitive ratio of $h$ if $k=h$ \cite{ChrobakKPV91,BartalK04} 
-- perform strictly worse in the resource augmentation setting than 
disabling the $k-h$ extra servers and applying the same algorithm to $h$ 
servers only. For the case that $h$ is not fixed, the Work Function 
Algorithm was shown to be $2h$-competitive simultaneously against any 
number $h\le k$ of offline servers \cite{Koutsoupias99}.

In terms of lower bounds, it is known that unlike for (weighted) paging, the 
competitive ratio does not converge to $1$ on general metrics even as 
$k/h\to\infty$. Prior to this work, the best known lower bounds were $2$ 
on the line \cite[p.\,175]{BorodinE98} and $2.4$ on general metric spaces 
\cite{BansalEJK17}.

The closest publication to this work is by Csirik et al.~\cite{CsirikINSW01}, which studies a problem that is essentially the special case of the infinite server problem on
the uniform metric space augmented by a far away source. It is cast as a paging
problem where new cache slots can be bought at a fixed price per unit
and gives matching upper and lower bounds of $\approx3.146$ on the
competitive ratio.

\subsection{Our Results}
Our main result is an equivalence theorem between the infinite server 
problem and the $(h,k)$-server problem, presented in 
Section~\ref{sec:equiv}. It 
states that the infinite server 
problem is competitive on every metric space if and only if the 
$(h,k)$-server problem is $O(1)$-competitive on every metric space as 
$k/h\to\infty$. We show further that it is not even necessary to let $k/h$ 
converge to infinity because in the positive case, there must also exist 
some $k=O(h)$ for which the latter is true. The theorem holds also if 
``every metric space'' is replaced by ``the real line''.

In Section~\ref{sec:Lbs} we present upper and lower bounds on the competitive 
ratio of the infinite server problem on a variety of metric spaces.  Extending 
the work in \cite{CsirikINSW01}, we present a tight lower bound for non-discrete spaces, which is then turned into a $3.146$ lower bound for the $(h,k)$ 
setting. To our knowledge, this is the largest bound on the weak 
adversaries setting for any metric space, as $k/h \rightarrow \infty$. We  
show how recent work by Bansal et al.~\cite{BansalEJK17} can be adapted to 
give an upper bound on the competitive ratio of the infinite server problem 
on bounded-depth weighted trees. We also consider
layered graph metrics, which are equivalent (up to a factor of 2) to general graph metrics. We have not settled the case for 
their competitive ratio, but we present a natural algorithm with tight analysis
and pose challenges for further research. The main open question is whether there exists a metric space on which the infinite server problem is not competitive.

In Section~\ref{sec:failed} we show how a variety of known
algorithms such as the work function and balancing algorithms fail for the infinite server problem, even on the real line. We focus in particular on a class of speed-adjusted variants of the well-known double coverage algorithm.

Finally, we present a useful reduction from arbitrary metric spaces to bounded subspaces in Section~\ref{sec:redBounded}. In particular, the infinite server problem on the line is competitive if and only if it is competitive for the special case where requests are restricted to some bounded interval further away from the source.

\subsection{Preliminaries}

Let $M=(M,d)$ be a metric space and let $s$ be a point of $M$. In the
\emph{infinite server problem on $(M,s)$}, an unbounded number of servers starts
at point $s$ and serves a finite sequence
$\sigma=(\sigma_0=s,\sigma_1,\sigma_2,\ldots, \sigma_m$) of requests
$\sigma_i\in M$. Serving a request entails moving one of the servers
to it. The goal is to minimize the total distance traveled by the
servers. 

We drop $s$ in the notation if the location of the source is not relevant or understood. We refer to the action of moving a server from the source to another point as \emph{spawning}. Throughout this work we use the letter $d$ for the metric
associated with the metric space.

In the online setting, the requests are revealed one by one and need to be served immediately without knowledge of future requests. All algorithms considered in this
paper are deterministic. An algorithm is called \emph{lazy} if it moves only one server to
serve a request at an unoccupied point and moves no server if the
requested point is already covered. An algorithm is called \emph{local} \cite{cohen2014pricing} if it moves a server from $a$ to $b$ only if there is no server at some other point $c$ on a shortest path from $a$ to $b$, \ie with $d(a,b)=d(a,c)+d(c,b)$. It is easy to see that any algorithm can be turned into a lazy and local algorithm without increasing its cost (\ie the total distance traveled by all servers).

For an algorithm $\ALG$, we denote by $\ALG(\sigma)$ its cost on the request sequence $\sigma$. Similarly, we write $\OPT(\sigma)$ for the optimal (offline) cost. 

An online algorithm $\ALG$ is \emph{$\rho$-competitive} for $\rho\ge 1$ 
if $\ALG(\sigma)\le\rho\OPT(\sigma)+c$ for all $\sigma$, where $c$ is a 
constant independent of $\sigma$. The \emph{competitive ratio} of an 
algorithm is the infimum of all such $\rho$. We say that an algorithm is 
\emph{competitive} if it is $\rho$-competitive for some $\rho$. We also 
call an online problem itself ($\rho$-)competitive if it admits such an 
algorithm. If the additive term $c$ in the definition is $0$, then the 
algorithm is also called \emph{strictly 
	$\rho$-competitive}~\cite{emek2009additive}.

The \emph{$(h, k)$-server problem on $M$} is defined like the infinite server problem except that the number of servers is $k$ for the online algorithm and $h$ for the optimal (offline) algorithm against whom it is compared in the definition of competitiveness. For this problem, the servers are not required to start at the same point, although a different initial configuration would only affect the additive term $c$. The problem is interesting only when $k\geq h$. The
case $h=k$ is the standard $k$-server problem and the case $k\geq h$
is known as the \emph{weak adversaries} model. One major open problem
is determine the competitive ratio of the $(h,k)$-server problem as
$k$ tends to infinity.

We will sometimes write $\OPT_h$ and $\OPT_\infty$ for the optimal offline algorithm, where the index specifies the number of servers available.

The following two propositions will be useful later in the paper.

\begin{proposition}\label{prop:commonCR}
	If for every metric space there exists a competitive algorithm for the
	infinite server problem, then there exists a universal competitive
	ratio $\rho$ such that the infinite server problem is \emph{strictly}
	$\rho$-competitive on every metric space.
\end{proposition}
\begin{proof}
	We first show the existence of $\rho$ such that the infinite server problem 
	is $\rho$-competitive (strictly or not) on every metric space. Suppose such 
	$\rho$ does not exist, then for every $n\in\mathbb{N}$ we can find a 
	metric space $M_n$ containing some point $s_n$ such that the infinite 
	server problem on $(M_n,s_n)$ is not $n$-competitive. Consider the metric 
	space obtained by taking the disjoint union of all spaces $M_n$ and gluing 
	all the points $s_n$ together. The infinite server problem would not be 
	competitive on this metric space, in contradiction to the assumption.
	
	Analogously we can also find a universal constant $c$ that works for
	all metric spaces as additive constant in the definition of
	$\rho$-competitiveness. A scaling argument shows that also $c=0$
	works.
\end{proof}

With a very similar argument we get:

\begin{proposition}\label{prop:commonCR-hk}
	Let $k=k(h)$ be a function of $h$. Suppose that for every metric space $M$ and for all $h$ there exists an $O(1)$-competitive algorithm for the
	$(h,k)$-server problem on $M$. Then there exists a universal competitive ratio
	$\rho$ such that the $(h,k)$-server problem is strictly $\rho$-competitive on
	every metric space if all servers start at the same
	point.
\end{proposition}

\section{Equivalence of Infinite Servers and Weak Adversaries}\label{sec:equiv}

The main result of this section is the following tight
connection between the infinite server problem and the weak
adversaries model.

\begin{theorem}\label{thm:equiv}
	The following are equivalent:
	\begin{enumerate}
		\item The infinite server problem is competitive.\label{it:equivInfComp}
		\item The $(h,k)$-server problem is $O(1)$-competitive as
		$k/h\to\infty$.\label{it:equivWeakComp}
		\item For each $h$ there exists $k=O(h)$ so that the $(h,k)$-server problem is $O(1)$-competitive.\label{it:equivkOhComp}
	\end{enumerate}
	The three statements above are also equivalent if we fix the metric space
	to be the real line.
\end{theorem}

The implication
``\ref{it:equivkOhComp}$\implies$\ref{it:equivWeakComp}'' is trivial.
The proof of the equivalence theorem consists in its core of two
reductions. Theorem~\ref{thm:reductionToWeak} contains the easier of
the two reductions, which is from the infinite server problem to the
$k$-server problem against weak adversaries
(``\ref{it:equivWeakComp}$\implies$\ref{it:equivInfComp}''). By
Propositions~\ref{prop:commonCR} and \ref{prop:commonCR-hk}, it suffices to
consider only strictly competitive
algorithms. Theorem~\ref{thm:redWeakToInf} proves essentially the
inverse for general metric spaces, and
Theorem~\ref{thm:redWeakInfLine} specializes it to the line
(``\ref{it:equivInfComp}$\implies$\ref{it:equivkOhComp}'').

As a corollary of the theorem we get the non-trivial implication
``\ref{it:equivWeakComp}$\implies$\ref{it:equivkOhComp}'', a
potentially useful statement towards resolving the major open problem
about weak adversaries: ``Is Statement \ref{it:equivWeakComp} true?''
This highlights the importance of the infinite server problem.

\begin{theorem}\label{thm:reductionToWeak}
	Fix a metric space $M$ and consider algorithms with all servers
	starting at some $s\in M$. If for every $h$ there exists $k=k(h)$ such
	that the $(h,k)$-server problem on $M$ is strictly $\rho$-competitive, for
	some constant $\rho$, then there exists a strictly $\rho$-competitive
	online strategy for the infinite server problem on $M$.
\end{theorem}
\begin{proof}
	Let $\ALG_{k(h)}$ denote an online algorithm with $k(h)$ servers that
	is strictly $\rho$-competitive against an optimal algorithm $\OPT_h$
	for $h$ servers, \ie
	\begin{align}\label{eq:redToWeak}
		\ALG_{k(h)}(\sigma)\leq \rho \OPT_h(\sigma)
	\end{align}
	for every request sequence $\sigma$. Without loss of generality,
	algorithm $ALG_{k(h)}$ is lazy.
	
	For every request sequence $\sigma$, consider the equivalence relation
	$\equiv_{\sigma}$ on natural numbers in which $h\equiv_{\sigma} h'$ if
	and only if $\ALG_{k(h)}(\sigma)$ and $\ALG_{k(h')}(\sigma)$ serve
	$\sigma$ in exactly the same way (\ie, make exactly the same moves).
	To every $\sigma$, we associate an equivalence class $H(\sigma)$ of
	$\equiv_{\sigma}$ that satisfies
	\begin{itemize}
		\item $H(\sigma)$ is infinite,
		\item $H(\sigma r)\subseteq H(\sigma)$, for every request $r$. 
	\end{itemize}
	This is done inductively in the length of $\sigma$ (in a manner
	reminiscent of K\"onig's lemma) as follows: For the
	base case when $\sigma$ is the empty request sequence,
	$H(\sigma)=\mathbb{N}$. For the induction step, suppose that we have
	defined $H(\sigma)$. Consider the equivalence classes of
	$\equiv_{\sigma r}$, a refinement of the equivalence classes of
	$\equiv_{\sigma}$. Since there are only finitely many possible ways
	to serve $r$, they partition $H(\sigma)$ into finitely many
	parts. At least one of these parts is infinite and we select it to be
	$H(\sigma r)$; if there is more than one such sets, we select one
	arbitrarily, say the lexicographically first.
	
	Given such a mapping $H$, we define the online algorithm
	$\ALG_{\infty}$ which serves every $\sigma$ in the same way as all the
	online algorithms $\ALG_{k(h)}$ for $h\in H(\sigma)$. The
	second property of $H$ guarantees that $\ALG_{\infty}$ is a
	well-defined online algorithm. 
	
	By construction, $\ALG_{\infty}(\sigma)=\ALG_{k(h)}(\sigma)$ for every
	$h\in H(\sigma)$. To finish the proof, observe that since $H(\sigma)$
	is infinite, it contains some $h$ greater than the length of $\sigma$,
	and for such an $h$ we have
	$OPT_\infty(\sigma)=OPT_h(\sigma)$. Substituting these to
	\eqref{eq:redToWeak}, we see that $\ALG_{\infty}$ is strictly
	$\rho$-competitive.
\end{proof}


We now show the reduction from the $k$-server problem against weak
adversaries to the infinite server problem on general metric spaces.

\begin{theorem}\label{thm:redWeakToInf}
	If the infinite server problem on general metric spaces is strictly $\tilde{\rho}$-competitive, then
	there exists a constant $\rho$ such that the $(h,k)$-server problem is
	$\rho$-competitive, for $k=O(h)$. In particular, for every
	$\epsilon>0$, we can take $\rho=(3+\epsilon)\tilde{\rho}$ and any
	$k\geq (1+1/\epsilon)\tilde{\rho} h$.
\end{theorem}
\begin{proof}
	Fix some metric space $M$ and a point $s\in M$. We will describe a
	strictly $\rho$-competitive algorithm for the $(h,k)$-server problem
	on $M$ for the case that all servers start at $s$. This implies a (not necessarily strictly) $\rho$-competitive
	algorithm for any initial configuration.
	
	The idea is to simulate a strictly $\tilde{\rho}$-competitive infinite
	server algorithm, but whenever it would spawn a $(k+1)$-st server, we
	bring all servers back to the origin and restart the
	algorithm. The problem is that the overhead cost for returning the
	servers to the origin, may be very high. To compensate for this, we
	assume that every time the servers return to the origin, they pretend
	to start from a different point further away from the origin. This
	motivates the following notation:
	
	\begin{definition}
		Given a metric $M$, a point $s\in M$, and a value $w\ge 0$, we
		will use the notation $M_{s\oplus w}$ to denote the metric derived
		from $M$ when we increase the distance of $s$ from every other point
		by $w$; we will also denote the relocated point by $s\oplus w$.
	\end{definition}
	
	Let $\ALG_{\infty}$ denote a strictly $\tilde{\rho}$-competitive
	online algorithm for the infinite server problem. We now define an online algorithm
	$\ALG_k$ for $k$ servers (all starting at $s$). We will make use of
	the notation $A(\sigma; s)$ to denote the cost of algorithm $A$ to
	serve the request sequence $\sigma$ when all servers start at
	$s$.
	
	\begin{definition}[$\ALG_{k}$ derived from $\ALG_{\infty}$]
		Algorithm $\ALG_k$ runs in phases with the initial phase 
		being the
		$0$th phase. At the beginning of every phase, all servers of $\ALG_k$
		are at $s$. In every phase $i$, the algorithm simulates the infinite
		server algorithm $\ALG_{\infty}$, whose servers start at $s\oplus w_i$
		for some $w_i\geq 0$. The parameters $w_i$ are determined online, and
		initially $w_0=0$.
		Whenever $\ALG_{\infty}$ spawns a server from $s\oplus w_i$, algorithm
		$\ALG_k$ spawns a server from $s$.
		
		The phase ends just before $\ALG_{\infty}$ spawns its
		$(k+1)$-st server or when the request sequence ends. In the former
		case, all servers of $\ALG_k$ return to $s$ to start the
		$(i+1)$-st phase. To determine the starting point of the simulated
		algorithm of the next phase, we set
		\begin{align}\label{eq:redDefWi}
			w_{i+1} = \epsilon \, \frac{\OPT_h(\sigma_{i}; s)}{h}\,,
		\end{align}
		where $\sigma_i$ is the sequence of requests during phase $i$.
	\end{definition}
	
	Let $n$ be the number of phases. The cost  of $\ALG_k$ for the requests 
	in 
	phase $i<n$ is $\ALG_{\infty}(\sigma_i; s\oplus w_i) - k w_i$; the last
	term is subtracted because the $k$ servers do not have to actually
	travel the distance between $s\oplus w_i$ and $s$. However for the
	last phase no such term can be subtracted since we do not know how
	many servers are spawned during the phase, and we can only bound the
	cost from above by $\ALG_{\infty}(\sigma_n; s\oplus w_n)$. The cost of
	returning the servers to $s$ at the end of a phase can at most
	double the cost during the phase.
	
	From this, we see that the total cost of $\ALG_k$ in phase $i$ is
	\begin{align*}
		\cost_i & \leq
		\begin{cases}
			2\left(\ALG_{\infty}(\sigma_i; s\oplus w_i)-k w_i\right) & \text{for
				$i<n$} \\
			\ALG_{\infty}(\sigma_n; s\oplus w_n) & \text{for $i=n$\,.}
		\end{cases}
	\end{align*}
	
	Since $\ALG_{\infty}$ is strictly $\tilde{\rho}$-competitive, we have
	\begin{align*}
		\ALG_{\infty}(\sigma_i; s\oplus w_i)
		&\le \tilde{\rho}\,\OPT_{\infty}(\sigma_i; s\oplus w_i)\\
		&\le \tilde{\rho}\,\OPT_h(\sigma_i; s\oplus w_i)\\
		&\le \tilde{\rho}\,(\OPT_h(\sigma_i; s)+hw_i)
	\end{align*}
	and substituting this in the expression for the cost, we can bound the
	total cost by
	\begin{align*}
		\ALG_k(\sigma; s) = \sum_{i=0}^n \cost_i 
		&\le 2\sum_{i=0}^{n-1}  \left( \tilde{\rho} (\OPT_h(\sigma_i; s) +h w_i)-k w_i\right)
		+\tilde{\rho} (\OPT_h(\sigma_n; s) +h w_n) \\
		&= 2\sum_{i=0}^{n-1}  \left( \tilde{\rho} \OPT_h(\sigma_i;
		s) -(k-\tilde{\rho}h)w_i\right)
		+\tilde{\rho}\OPT_h(\sigma_n; s) +\tilde{\rho}h w_n\,.
	\end{align*}
	
	The parameters $w_i$ and $k$ were selected so that the summation
	telescopes, and we are left with
	\belowdisplayskip=0pt
	\begin{align*}
		\ALG_k(\sigma; s) & \leq  2\,\tilde{\rho}\, \OPT_h(\sigma_{n-1}; s)+
		\tilde{\rho} \, \OPT_h(\sigma_{n}; s) + \tilde{\rho}\, \epsilon\,
		\OPT_h(\sigma_{n-1}; s)\\
		& \leq (3+\epsilon)\, \tilde{\rho}\, \OPT_h(\sigma; s)\,.
	\end{align*}
\end{proof}

The previous reduction requires the infinite server problem to be
competitive on \emph{every} metric space. The following variant only
requires the infinite server problem to be competitive on the line.
\begin{theorem}\label{thm:redWeakInfLine}
	If the infinite server problem on the line is $\rho$-competitive, then
	for every $h\in\mathbb{N}$ and $\epsilon>0$, the $(h,k)$-server 
	problem on
	the line is $(3+\epsilon)\rho$-competitive, when
	$k\geq 2\lceil (1+1/\epsilon)\rho h\rceil$.
\end{theorem}
\begin{proof}
	A straightforward adaptation of the proof of the previous lemma, shows
	the existence of a $(3+\epsilon)\rho$-competitive algorithm for the
	interval $[0,\infty)$, when $k\geq 2(1+1/\epsilon)\rho h$. By doubling
	the number of online servers so that half of them are used in each
	half-line, we get a $(3+\epsilon)\rho$-competitive algorithm for the
	entire line, when $k\geq 2\lceil (1+1/\epsilon)\rho h\rceil$.
	
	Note that the proof assumes strictly competitive algorithms.  But, by
	a straightforward scaling argument, if the infinite server problem on
	the line is $\rho$-competitive, then it is also strictly
	$\rho$-competitive. This in turn implies a strictly $\rho$-competitive
	online algorithm for $M_{0\oplus w}$, since this space is isometric to
	the subspace $\{-w\}\cup(0,\infty)$ of the line. 
\end{proof}

In the next section we look at some particular metric spaces and give upper and lower bounds on the competitive ratio.

\section{Upper and Lower Bounds}\label{sec:Lbs}

Unlike the $k$-server problem, which is $1$-competitive if and only if the metric spaces has at most $k$ points and conjectured $k$-competitive otherwise, the situation is more diverse for the infinite server problem. For example, on uniform metric spaces (where all distances are the same) the problem is trivially $1$-competitive even if the metric space consists of uncountably many points. This is because an optimal strategy in this case is to spawn a server to every requested point. More generally, this strategy achieves a finite competitive ratio on any metric space where distances are bounded from below and above by positive constants. This suggests that statements about the competitive ratio for the infinite server problem cannot be as simple as the (conjectured) dichotomy for the $k$-server problem, which depends only on the number of points of the metric space. In this section we derive bounds on the competitive ratio for particular classes of metric spaces.

\subsection{Weighted Trees}
We consider the infinite server problem on metric spaces that can be modeled by edge-weighted trees. The points of the metric space are the nodes of the tree, and the distance between two nodes is the sum of edge weights along their connecting path. We choose the source of the metric space as the root of the tree, and define the depth of the tree as the maximal number of edges from the root to a leaf. The number of nodes can be infinite (otherwise the infinite server problem is trivially $1$-competitive), but we assume the depth to be finite.

An upper bound on the competitive ratio of such trees follows easily from an upper bound for the $(h,k)$-server on such trees \cite{BansalEJK17} and the equivalence theorem:

\begin{theorem}
	The competitive ratio of the infinite server problem on trees of depth $d$ is at most $O(2^d\cdot d)$.
\end{theorem}
\begin{proof}
	Bansal et al.~\cite[Theorem 1.3]{BansalEJK17} showed that the competitive ratio of
	the $(h,k)$-server problem on trees of depth $d$ is at most
	$O(2^d\cdot d)$ provided that $k/h$ is large enough. Inspection of the proof in
	\cite{BansalEJK17} shows that if all servers start at the root, it is in fact strictly $O(2^d\cdot d)$-competitive. Thus, Theorem~\ref{thm:reductionToWeak} implies the result for the infinite server problem.
\end{proof}


\subsection{Non-Discrete Spaces and Spaces with Small Infinite Subspaces}

The following theorem gives a lower bound of $3.146$ on the competitive ratio of the infinite server problem on any metric space containing an infinite subspace of a diameter that is small compared to the subspace's distance from the source. For example, every non-discrete metric space has this property (unless the source is the only non-discrete point), since non-discrete metric spaces contain infinite subspaces of arbitrarily small diameter. The theorem is a generalization of such a lower bound established in~\cite{CsirikINSW01} for a variant of the paging problem where cache cells can be bought. Crucial parts of the subsequent proof are as in~\cite{CsirikINSW01}.

\begin{theorem}\label{thm:3.146Lb}
	Let $M$ be a metric space containing an infinite subspace $M_0\subset M$ of finite diameter $\delta$ and a point $s\in M\setminus M_0$ such that the infimum $\Delta$ of distances between $s$ and points in $M_0$ is positive. Let $\lambda> 3.146$ be the largest real solution to
	\begin{align}
		\lambda=2+\ln \lambda\,.\label{eq:def3.146}
	\end{align}
	The competitive ratio of any deterministic online algorithm for the infinite server problem on $(M,s)$ is bounded from below by a value that converges to $\lambda$ as $\Delta/\delta\to\infty$. In particular, the competitive ratio is at least $\lambda$ if $M\setminus\{s\}$ contains a non-discrete part.
\end{theorem}
\begin{proof}
	By scaling the metric, we can assume that $\delta=1$. Let $p_1,p_2,p_3,\dots$ be infinitely many distinct points in $M_0$.
	
	Fix some lazy deterministic online algorithm $\ALG$. We consider the request sequence that always requests the point $p_i$ with $i$ minimal such that $p_i$ is not occupied by a server of $\ALG$. We call a move of a server between two points in $M_0$ \emph{local} (\ie every move that does not spawn is local). Let $f_j$ be the cumulative cost of local moves incurred to $\ALG$ until it spawns its $j$th server. Let $\sigma_k$ be this request sequence that is stopped right after $\ALG$ spawns its $k$th server, for some large $k$. The total online cost is
	\begin{align}
		\ALG(\sigma_k)\ge k\Delta + f_k\,.\label{eq:3.146cost}
	\end{align}
	
	Let $h=\lceil k/\lambda\rceil$. We consider several offline algorithms 
	that start behaving the same way, so we think of it as one algorithm 
	initially that is forked into several algorithms later. The offline 
	algorithms make use of only $h$ servers and they begin by spawning 
	them to the points $p_1,\dots,p_h$. They do not need to move any 
	servers until $\ALG$ spawns its $h$th server. Whenever $\ALG$ spawns its 
	$j$th server for some $j\ge h$, every offline algorithm is forked to $h$ 
	distinct algorithms: Each of them moves a different server to $p_{j+1}$ 
	(to prepare for the next request, which will be at $p_{j+1}$). We will keep 
	the invariant that each offline algorithm already has a server at the next 
	request. To this end, whenever $\ALG$ does a local move from $p$ to $p'$, 
	every offline algorithm that does not have a server at $p$ moves a server 
	from $p'$ to $p$; note that the algorithm had a server at $p'$ by the 
	invariant, and the next request will be at $p$.
	
	When $\ALG$ has $j$ spawned servers ($j\ge h$), the offline algorithms are 
	in ${{j} \choose {h-1}}$ different configurations, each of which occurs 
	equally often among them. If $\ALG$ does a local move from $p$ to $p'$, 
	there are ${{j-1}\choose {h-1}}$ different offline configurations for which 
	a local move is made in the opposite direction. Thus, for each local move 
	by $\ALG$ while having $j$ servers in total, a portion ${j-1\choose h-1}/{j\choose h-1} = \frac{j-h+1}{j}$
	of the offline algorithms move a server in the opposite direction for the same cost.
	
	We use the average cost of all offline algorithms we considered as an upper bound on the optimal cost. The cost of spawning $h$ servers is at most $h(\Delta+1)$, and the average cost while $\ALG$ has $j$ spawned servers (for $j=h,\dots,k-1$) is at most $\frac{j-h+1}{j}(f_{j+1}-f_j)+1$ (with the ``$+1$'' coming from the move when offline algorithms fork). Hence,
	\begin{align}
		\OPT(\sigma_k) &\le h(\Delta+1) + k-h + \sum_{j=h}^{k-1}\frac{j-h+1}{j}(f_{j+1}-f_j)\,,\nonumber\\
		&\le h\Delta + k + \frac{k-h}{k-1}f_k - \frac{f_h}{h} - \sum_{j=h+1}^{k-1}\frac{h-1}{j(j-1)}f_{j}\,,\nonumber
	\end{align}
	Note that $\frac{f_k}{k}$ is bounded from above because otherwise $\ALG$ would not be competitive, and it is bounded from below by $0$. Thus, $L=\liminf_{k\to\infty}\frac{f_k}{k}$ exists. In the following we use the asymptotic notation $o(1)$ for terms that disappear as $k\to\infty$. We can choose arbitrarily large values of $k$ such that $\frac{f_k}{k}= L+o(1)$. Since $h=\lceil k/\lambda\rceil$, we have $\frac{f_j}{j}\ge L+o(1)$ for all $j\ge h$. Moreover, $\sum_{j=h+1}^{k-1}\frac{1}{j-1}=\ln(\lambda)+o(1)$. This allows us to simplify the previous bound to
	\begin{align*}
		\OPT(\sigma_k) &\le \frac{k}{\lambda}\Bigl(\Delta + \lambda + \bigl(\lambda-1 - \ln(\lambda)\bigr)L + o(1)\Bigr)\\
		&= \frac{k}{\lambda}\bigl(\Delta + L + \lambda + o(1)\bigr)\,,\nonumber
	\end{align*}
	where the last step uses equation \eqref{eq:def3.146}.
	
	The competitive ratio is at least
	\begin{align*}
		\frac{\ALG(\sigma_k)+O(1)}{\OPT(\sigma_k)} &\ge \frac{k\Delta+f_k+O(1)}{\frac{k}{\lambda}\bigl(\Delta + L + \lambda + o(1)\bigr)}\\
		&= \lambda\cdot\frac{\Delta+L}{\Delta + L + \lambda} + o(1)\,.
	\end{align*}
	The fraction in the last term tends to $1$ as $\Delta\to\infty$.
\end{proof}

This bound is tight due to a matching upper bound in \cite{CsirikINSW01} that shows (translated to the terminology of the infinite server problem) that a competitive ratio of $\lambda$ can be achieved on metric spaces where all pairwise distances are $1$ except that the source is at some larger distance $\Delta$ from the other points.

The previous theorem together with the equivalence theorem also allows us to obtain a new lower bound for the $k$-server problem against weak adversaries.

\begin{corollary}
	For sufficiently large $h$, there is no $3.146$-competitive algorithm for the $(h,k)$-server problem on the line, even if $k\to\infty$.
\end{corollary}
\begin{proof}
	By a scaling argument it is easy to see that if the infinite server problem on the line is $\rho$-competitive, then it is also \emph{strictly} $\rho$-competitive. Thus, the statement follows from Theorems~\ref{thm:reductionToWeak} and \ref{thm:3.146Lb}.
\end{proof}

This improves upon both the previous best known lower bounds of $2$ for this problem on the line \cite[p.\,175]{BorodinE98} and $2.4$ on general metric spaces \cite{BansalEJK17}.

\subsection{Layered Graphs}\label{sec:layered}

A \emph{layered graph of depth $D$} is a graph whose (potentially infinitely many) nodes can be arranged in layers $0,1,\dots,D$ so that all edges run between adjacent layers and each node -- except for a single node in layer $0$ -- is connected to at least one node of the previous layer. The induced metric space is the set of nodes with the distance being the minimal number of edges of a connecting path. For the purposes of the infinite server problem, the single node in layer $0$ is the source. We assume $D\ge 2$ to avoid trivial cases.

Note that a connected graph is layered if and only if it is bipartite. Moreover, any graph can be embedded into a bipartite graph by adding a new node in the middle of each edge. So essentially, layered graphs capture \emph{all} graph metrics.

Let \emph{Move Only Outwards (MOO)} be some lazy and local algorithm 
for the infinite server problem on layered graphs that moves servers along edges only in the direction away from the source. Not surprisingly, the competitive ratio 
of this simple algorithm is quite bad and we show that it is exactly 
$D-1/2$. Nonetheless, at least for $D\le3$ this is actually the optimal 
competitive ratio.

\begin{theorem}\label{thm:MOO}
	The competitive ratio of MOO is exactly $D-\frac{1}{2}$.
\end{theorem}
\begin{proof}
	\ \\
	\textit{Upper bound:\\}
	Consider some final configuration of the algorithm. Let $n_j$ be the 
	number of servers in the $j$th layer. Then the cost of MOO is
	\begin{align*}
		\cost=\sum_{j=1}^D jn_j.
	\end{align*}
	To obtain an upper bound on $\OPT$, observe that every node occupied 
	by MOO in the final configuration must have been visited by an offline 
	server at least once. We account an offline cost of $1$ for each visit of a 
	node on layers $1,\dots,D-2$ and an offline cost of $2$ for each visit of 
	a node on layer $D$. This cost of $2$ covers the last two edge-traversals 
	before visiting the layer-$D$-node, so this may include serving a 
	request on layer $D-1$. If $n_{D-1}>n_D$, then we can account another 
	$n_{D-1}-n_D$ cost for visiting the remaining at least $n_{D-1}-n_D$ 
	requested nodes on layer $D-1$. In summary,
	\begin{align*}
		\OPT \ge \sum_{j=1}^{D-2} n_j + 2n_D+(n_{D-1}-n_D)^+
	\end{align*}
	where $(n_{D-1}-n_D)^+:=\max\{0,n_{D-1}-n_D)$. The upper bound on 
	the competitive ratio follows since
	\begin{align*}
		\frac{\cost}{\OPT}
		&\le \frac{\sum_{j=1}^D jn_j}{\sum_{j=1}^{D-2} n_j + 
			2n_D+(n_{D-1}-n_D)^+}\\
		&\le \frac{(D-2)\sum_{j=1}^{D-2} n_j + (2D-1)n_D + 
			(D-1)(n_{D-1}-n_D)^+}{\sum_{j=1}^{D-2} n_j + 
			2n_D+(n_{D-1}-n_D)^+}\\
		&\le D-\frac{1}{2}\,.
	\end{align*}
	\ \\
	\textit{Lower bound:\\}
	Let $k,n\in\mathbb{N}$ be some large integers. We construct the 
	following graph: Layers $0,\dots,D-2$ consist of one node each and 
	layers $D-1$ and $D$ consist of infinitely many nodes each, denoted 
	$a_0,a_1,a_2,\dots$ and $b_0, b_1,b_2,\dots$ respectively. For each 
	$i\in\mathbb{N}_0$, the $k$ nodes $b_{ik},b_{ik+1}\dots,b_{(i+1)k-1}$ 
	are adjacent to each of the $2k$ nodes $a_{ik},a_{ik+1},a_{(i+2)k-1}$ and 
	to no other nodes. The set of remaining edges is uniquely determined by 
	the fact that this is a layered graph of depth $D$.
	
	The request sequence consists of $n$ rounds $0,1,\dots,n-1$, where 
	each request in round $i$ is at a node from the list 
	$a_{ik},a_{ik+1},\dots,a_{(i+1)k-1}, b_{ik},b_{ik+1},\dots,b_{(i+1)k-1}$. 
	Round $i$ starts with requests on the nodes 
	$a_{ik},a_{ik+1},\dots,a_{(i+1)k-1}$. Then, for $j=0,\dots,k-1$, the 
	adversary first requests $b_{ik+j}$ and then requests whichever node 
	from $a_{ik},a_{ik+1},\dots,a_{(i+1)k-1}$ has been left by an MOO-server 
	to serve the request at $b_{ik+j}$. Note that by definition of MOO and the 
	graph, the server it moves to $b_{ik+j}$ does indeed come from 
	$a_{ik},a_{ik+1},\dots,a_{(i+1)k-1}$.
	
	In round $i$, MOO first pays $k(D-1)$ to move $k$ servers to 
	$a_{ik},a_{ik+1},\dots,a_{(i+1)k-1}$ and then, for each $j=0,\dots,k-1$, 
	it pays $1$ to move to $b_{ik+j}$ and $D-1$ to spawn a new server at 
	the group $a_{ik},a_{ik+1},\dots,a_{(i+1)k-1}$. Over $n$ rounds this 
	makes a total cost of $n(k(D-1)+k(1+D-1))=nk(2D-1)$.
	
	The offline algorithm can serve requests as follows: The requests at 
	$a_{ik},\dots,a_{(i+1)k-1}$ at the beginning of round $i$ are served by 
	spawning if $i=0$ (for cost $(d-1)k$) and by sending servers from 
	$b_{(i-1)k},\dots,b_{ik-1}$ if $i\ge1$ (for cost $k$). The request at 
	$b_{ik}$ is served by spawning a server (cost $D$) and the requests at 
	$b_{ik+1},\dots,b_{ik+k-1}$ are served by sending a server from a node 
	in $a_{ik},\dots,a_{(i+1)k-1}$ that will not be requested any more (cost 
	$1$ each, so $k-1$ per round). Over $n$ rounds, this adds up to an 
	offline cost of $(D-1)k+(n-1)k+n(D+k-1)=2nk+(D-2)k+n(D-1)$. The 
	ratio of online and offline cost is
	\begin{align*}
		\frac{nk(2D-1)}{2nk+(D-2)k+n(D-1)}=\frac{2D-1}{2+\frac{D-2}{n}+\frac{D-1}{k}}\,,
	\end{align*}
	which gets arbitrarily close to $D-\frac{1}{2}$ for $n$ and $k$ large 
	enough.
\end{proof}

\begin{theorem}\label{thm:layeredLBs}
	The competitive ratio of the infinite server problem on layered graphs of depth $D$ is exactly $1.5$ for $D=2$, exactly $2.5$ for $D=3$ and at least $3$ for $D\ge 4$.
\end{theorem}
\begin{proof}
	For $D=2$, the only possibility to move a server closer to the source is 
	from layer $2$ to layer $1$. But since spawning to layer $1$ is at least as 
	good, we can restrict our attention to algorithms of the type MOO. The 
	result follows from Theorem~\ref{thm:MOO}.
	
	For $D=3$, the upper bound follows from Theorem~\ref{thm:MOO}. It 
	remains to show the lower bounds for $D\ge 3$.
	
	Fix some large integers $k,n\in\mathbb{N}$. Consider the following 
	layered graph of depth $D$. For $i=0,\dots,D-1$ there exists a node 
	$v_i$ in layer $i$. The remaining nodes are defined inductively as all 
	nodes obtained by the following two rules:
	\begin{itemize}
		\item There exist a set $S_0$ of $2k$ nodes and sets $A^{S_0}$ and 
		$B^{S_0}$ of $k$ nodes.
		\item Let $S$ be a set of $2k$ nodes such that $A^S$ and $B^S$ 
		exist. Then for each $S'\subset S\cup B^S$ of size $2k$ there are sets 
		$A^{S'}$ and $B^{S'}$ of $k$ nodes.
	\end{itemize}
	The nodes in the sets $A^S$ are in layer $D-1$, the nodes in $S_0$ and 
	in the sets $B^S$ are in layer $D$. For a node in some set $A^S$, the set 
	of adjacent nodes in layer $D$ is $S\cup B^S$. The remaining edges are 
	so that this is a layered graph with the layers as specified.
	
	For purposes of the analysis below, we further define a 
	\emph{generation} of a node as follows: The nodes $v_0,\dots,v_{D-1}$ 
	and the nodes in $S_0$ have generation $1$. The generation of nodes in 
	$A^S$ and $B^S$ is the maximal generation of any node in $S$ plus $1$.
	
	Let $\ALG$ be some online algorithm. We assume without loss of 
	generality that $\ALG$ is lazy and local.
	
	The adversary chooses the following request sequence against $\ALG$. 
	First, request the nodes in $S_0$ until $\ALG$ has a server at each of 
	them. The adversary also moves $2k$ servers towards these nodes. The 
	adversary uses only these $2k$ servers for the entire sequence of 
	requests. The remainder of the requests consists of several rounds. We 
	will keep the invariant that at the beginning of the $i$th round, the $2k$ 
	adversary servers occupy a set $S$ for which $A^S$ and $B^S$ (with 
	nodes of generation $i+1$) exist, and the online servers occupy nodes of 
	generation at most $i$. Clearly this holds before the first round. Let 
	$A^{S}=\{a_1,\dots,a_k\}$ and $B^{S}=\{b_1,\dots,b_k\}$.
	
	The requests of the $i$th round are divided into part a and part b, 
	consisting of \emph{steps} a.$1$,\dots,a.$k$, b.$1$,\dots,b.$k$ that are 
	executed in this order. Step a.$j$ consists of the following one or two 
	requests: First request $a_j$. If $\ALG$ moves a server from some $b\in 
	S$ towards $a_j$, immediately request $b$. We can assume that online 
	servers cover $A^{S}$ after the end of part a (otherwise request nodes in 
	$A^{S}$ again at the end of part a until this is the case). Step b.$j$ 
	consists of the following two or three requests: First request $b_{j}$. 
	Note that any path from a node of generation at most $i$ to $b_j$ 
	contains a node in $S$, and from any node in $S$, the shortest paths to 
	$b_j$ include the ones along the nodes in $A_S$. Thus, since $\ALG$ is 
	local, it will move a server from some $a\in A^{S}$ towards $b_j$. The 
	second request of step b.$j$ is at this node $a$ and, if $\ALG$ moves a 
	server from some $b\in S\cup B^S$ towards $a$, then the step contains 
	a third request at $b$.
	
	The adversary cost per round is at most $2k+2$: For each $j=1,\dots,k$, 
	there are at least $j$ nodes in $S$ that will \emph{not} be requested 
	during steps a.$j$, \dots, a.$k$, b.$1$, \dots, b.$(k-1)$. Hence, the 
	adversary can serve all requests of part a for cost $k$ by moving $k$ 
	servers from $S$ towards $A^{S}$ whilst keeping servers at all nodes of 
	$S$ that will be requested during the steps b.$1$,\dots,b.$(k-1)$. 
	Similarly, it can serve the steps b.$1$,\dots,b.$(k-1)$ for cost $k-1$ by 
	moving $k-1$ servers from $A^{S}$ to $B^S$. The final step b.$k$ of the 
	round can be served at cost $3$ using the last server in $A^{S}$ to serve 
	the requests and finish with all $2k$ offline servers in some set 
	$S'\subseteq S\cup B^S$.
	
	We analyze the online cost for the cases $D=3$ and $D\ge 4$ separately.
	
	If $D=3$, then the cost for each step a.$j$ is at least $2$ and the cost 
	for each step b.$j$ is at least $3$. Thus, the cost per round is at least 
	$5k$. As $k$ goes to infinity, the ratio of online and offline cost in each 
	round converges to $2.5$. As the number of rounds goes to infinity, the 
	online and offline costs before the first round become negligible, which 
	proves the lower bound of $2.5$ for $D=3$.
	
	For $D\ge 4$, we use a potential $\Phi$ equal to the number of online 
	servers in layer $D-1$. During step a.$j$, either $\Phi$ does not change 
	and the cost is at least $2$, or $\Phi$ increases by $1$ and the cost is at 
	least $3$. Thus, during step a.$j$ we have $\Delta\cost \ge 
	2+\Delta\Phi$ and hence during part a we have $\Delta\cost \ge 
	2k+\Delta\Phi$. During step b.$j$, either $\Phi$ decreases by $1$ and 
	the cost is at least $3$, or $\Phi$ does not change and the cost is at 
	least $4$. Thus, during part b we have $\Delta\cost \ge 4k+\Delta\Phi$. 
	In total, this adds up to $\Delta\cost \ge 6k+\Delta\Phi$ during the 
	round. Over $n$ rounds, this makes $\Delta\cost \ge 
	6nk+\Delta\Phi\ge 6nk$ since $\Phi$ starts at $0$ before the first round 
	and remains nonnegative. As $k$ and $n$ go to infinity, the ratio of our 
	bounds on online and offline cost converges to $3$.
\end{proof}

It remains an open problem to close the gap between the lower bound of $3$ and the upper bound of $3.5$ for $D=4$. More importantly, we are interested in the question whether an algorithm better than MOO exists for large $D$, achieving a competitive ratio of less than $D-1/2$ on any layered graph of depth $D$. Note that if no algorithm with a competitive ratio of $O(1)$ as $D\to\infty$ exists, then the infinite server problem on general metric spaces would not be competitive.

For large $D$, the lower bound of $3$ is certainly not tight: Consider a layered graph where each layer contains one node except that the bottom layer contains infinitely many nodes. By Theorem~\ref{thm:3.146Lb} (and a matching upper bound shown in \cite{CsirikINSW01}), the competitive ratio on this graph converges to $\lambda\approx3.146$ as $D\to\infty$.

\section{Algorithms with Unbounded Competitive Ratio}\label{sec:failed}

We examine the performance of classical algorithms known for the $k$-server problem when applied to the infinite server problem. The main focus of this section is a generalization of the Double Coverage algorithm for the line with adjusted server speeds. This idea has proved successful for the $(h,k)$-server problem (and hence the infinite server problem) on weighted trees \cite{BansalEJK17}. However, neither of these algorithms is competitive for the infinite server problem even on the line.

\subsection{Work Function Algorithm}
The Work Function Algorithm (WFA, \cite{chrobak1992server}) for the $k$-server problem achieves a 
competitive ratio of at most $2k-1$, which is the best known upper bound 
for general metric spaces \cite{KoutsoupiasP95}. Given a sequence of 
requests $r_1,r_2,\dots$ and a configuration $C$ (\ie a multiset of server 
positions), the work function $w_t(C)$ is defined as the minimal cost of 
serving the first $t$ requests and ending up in configuration $C$. If 
$C_{t-1}$ is the server configuration before the $t$th request, the 
algorithm moves to a configuration $C_t$ that contains $r_t$ and 
minimizes the quantity \begin{align}
w_t(C_t)+d(C_{t-1},C_t)\,,\label{eq:WFADef}
\end{align}
where $d(C_{t-1},C_t)$ is the cost of moving from $C_{t-1}$ to $C_t$.

\begin{proposition}
	The WFA is not competitive for the infinite server problem on the line.
\end{proposition}
\begin{proof}
	Let the source be at $0$ and let $p_1,p_2,\dots$ be infinitely many points in the interval $[1,1+\delta]$ for some small $\delta>0$. Consider the request sequence that always requests the point $p_i$ 
	with $i$ minimal such that $p_i$ is not occupied by an online server. Let $\sigma_k$ be the prefix of this request sequence until the WFA spawns its $k$th server.
	It is easy to see that the WFA spawns its $k$th server only if the optimal way of serving the already seen requests is to bring $k$ servers to the points $p_1,\dots,p_k$. In particular, $\OPT(\sigma_k)= k+o(1)$ as $\delta\to0$. Thus, the optimal cost increases by $1+o(1)$ during the period when the WFA has $k$ spawned servers, and the same is true for the optimal cost of an offline algorithm that is restricted to using $k$ servers only. Let $\cost_k$ be the cost incurred to the WFA during this period. Due to the lower bound of $k$ on the competitive ratio of any $k$-server algorithm, $\cost_k$ is at least $k$ times this increase of the optimal cost (up to an additive error of order $o(1)$ as $\delta\to0$), \ie $\cost_k\ge k+o(1)$. Thus, the total cost of WFA given the request sequence $\sigma_n$ is at least
	\begin{align*}
	\sum_{k=1}^{n-1}\cost_k = \Omega(n^2).
	\end{align*}
	Meanwhile, the optimal cost is $\OPT(\sigma_n)=n+o(1)$. Letting $n$ tend to infinity we obtain an unbounded competitive ratio.
\end{proof}

\subsection{Balance and Balance2}
The algorithm \emph{Balance} serves a request $r$ by sending a server $x$ 
that minimizes the quantity $D_x+d(x,r)$, where $D_x$ is the cumulative 
distance traveled by $x$ so far and $d(x,r)$ is the distance between $x$ and 
$r$. For the $k$-server problem, Balance is $k$-competitive on metric 
spaces with $k+1$ points \cite{ManasseMS88} and for weighted paging 
\cite{ChrobakKPV91}. Young showed that for weighted paging against a 
weak adversary with $h$ servers the competitive ratio of Balance is 
$k/(k-h+1)$ \cite{Young94}. On general metric spaces however, Balance 
has unbounded competitive ratio, even if $k=2$ \cite{ManasseMS88}. It is 
therefore unsurprising that it is also not competitive for the infinite server 
problem.

\begin{proposition}
	Balance is not competitive for the infinite server problem on 
	the line.
\end{proposition}
\begin{proof}
	Suppose all servers start at source $0$ and consider the request 
	sequence $r_0,r_1,r_2,\dots,r_n$ where $r_i=1-i\epsilon$. As 
	$\epsilon\to0$, the optimal cost tends to $1$ whereas the cost of 
	Balance tends to $n+1$. Since $n$ can be arbitrarily high, this shows an 
	unbounded competitive ratio.
\end{proof}

The intuitive problem of Balance is that it is not greedy enough. The 
algorithm \emph{Balance2} by Irani and Rubinfeld \cite{IraniR91} 
compensates for this weakness by giving more weight to the distance 
between the server and the request: To serve request $r$, Balance2 sends a 
server $x$ that minimizes the quantity $D_x+2d(x,r)$. Irani and Rubinfeld 
showed that, unlike Balance, Balance2 is competitive for two servers 
(achieving a competitive ratio of at most $10$) and they conjectured that it 
is also competitive for any other finite number of servers \cite{IraniR91}.

However, for the infinite server problem this algorithm is also not 
competitive:

\begin{proposition}
	Balance2 is not competitive for the infinite server problem on the line.
\end{proposition}
\begin{proof}
	Suppose the source is at $0$ and fix some small constant $\epsilon>0$. 
	The request sequence consists of several phases, starting with phase 
	$0$. Phase $i$ consists of alternating requests at $1-2i\epsilon$ and 
	$1-(2i+1)\epsilon$. We will ensure that all requests of a phase are 
	served by the same online server, and we call this the active server. As 
	soon as the cumulative distance traveled by the active server exceeds 
	$2-(4i+5)\epsilon$, the phase ends and a new phase begins. Note that 
	this means that the active server of a phase will not be used to serve any 
	request of a subsequent phase because, by definition of Balance2, the 
	algorithm would rather spawn a new server. Thus, the first request of 
	each phase is served by spawning a new server, which becomes the 
	active server of that phase. While the cumulative distance of the active 
	server is at most $2-(4i+5)\epsilon$ and since its distance from the next 
	request of the phase is always exactly $\epsilon$, its associated quantity 
	$D_s+2d(s,r)$ is at most $2-(4i+3)\epsilon$. Hence, Balance2 rather 
	uses this server during the phase instead of spawning a new server. 
	Thus, it is indeed the active server that serves \emph{all} requests of its 
	phase.
	
	Let $n$ be the number of phases and choose $\epsilon$ small enough 
	so that all requests are in the interval $[1/2,1]$. Thus, the cost of 
	Balance2 is $\Omega(n)$.
	
	An offline algorithm could serve all requests with two servers only that 
	move to $1$ and $1-\epsilon$ initially and then back towards $1/2$, 
	always covering the two points that are requested during a phase, 
	resulting in an offline cost of less than $3$. As $n$ goes to infinity, the 
	ratio between online and offline cost becomes arbitrarily large.
\end{proof}

\subsection{Double Coverage Variants}
Perhaps more surprising than for WFA and balancing algorithms is that a class of algorithms 
extending the Double Coverage (DC) algorithm \cite{ChrobakKPV91} is also not competitive for the infinite server problem. The basic DC algorithm on the line serves each request by an adjacent server. If the request lies between two servers, both servers move towards it at equal speed until one of them reaches the request. A sensible extension of this algorithm seems to be to give different speeds to 
servers, so that they move away from the source faster than towards it.

We consider here only the half-line $[0,\infty)$ with the source at the left border 0. Let $x_i$ be the position of the $i$th server 
from the right. We use the notation $x_i$ both for its position and 
for the server itself.  As servers do not overtake each other, $x_i$ is the 
$i$th spawned server. Let
$\mathcal{S} = \{s_i \ge 1
\mid i \in \mathbb{N}\text{ and } i \ge 2\}$ for a monotonic (non-decreasing or non-increasing) sequence of speeds $s_i$. The algorithm $\mathcal{S}$-DC is defined as follows:
\begin{itemize}
	\item
	If there exist servers $x_{i+1}$ and $x_i$ to the left and right of the
	request, move them towards it with speeds $s_{i+1}$ and 1
	respectively until one of the two reaches it.
	\item
	If a request does not have a server to its right, move the rightmost
	server to the request.
\end{itemize}
If $s_i=1$ for all $i$, this is precisely the original DC algorithm.

We will prove that $\mathcal S$-DC is not competitive. The intuitive reason is that servers move to the right either too slowly or too quickly: Imagine 
repeatedly requesting the same $n$ points in some small interval away from the source, until $\mathcal S$-DC covers all $n$ points. One case is that $\mathcal S$-DC spawns too slowly and is therefore defeated by an adversary covering these $n$ 
positions immediately with $n$ servers. In the other case, the adversary will also use $n$ servers to cover the initial group of requests and then shift its group of servers slowly towards the source, always making requests at the new positions of these offline servers. As $\mathcal S$-DC tries to cover the new requests, it is tricked into spawning 
too many servers. Both 
cases lead to an unbounded competitive ratio.

The proof consists of several lemmas. The lemmas hold also for non-monotonic speeds and we use monotonicity only to easily combine the lemmas in the end.

A useful property of $\mathcal S$-DC is that its cost can be calculated using only the final positions of the servers.

\begin{lemma} \label{lemma:online_cost}
	Let $x_1\ge x_2 \ge\dots$ be the server positions of $\mathcal{S}$-DC 
	after serving a sequence of requests. Then the cost paid is 
	$\sum_{i=1}^\infty z_i x_i$ where
	\begin{align}
		z_1 &= 1\\
		z_i &= \frac{z_{i-1}}{s_i} + 1 + \frac{1}{s_i}\,.
	\end{align}
\end{lemma}
\begin{proof}
	
	The movement $x_i$ of each server can be written as $x_i = r_i - l_i$ 
	where $r_i$ and $l_i$ are the cumulative distances traveled by that server 
	while moving to the right and left respectively.
	By definition of $\mathcal{S}$-DC, for all $i$ 
	we have
	\begin{align*}
		l_i = \frac{r_{i+1}}{s_{i+1}},
	\end{align*}
	since any right move (apart from the rightmost server) is accompanied by 
	a left move of another server.
	Observe that the online cost is
	\begin{align}
		\cost 
		&= \sum_{i=1}^\infty (r_i + l_i) = \sum_{i=1}^\infty (r_i + 
		\frac{r_{i+1}}{s_{i+1}}) 
		\nonumber\\
		&= \sum_{i=1}^\infty r_i + \sum_{i=2}^\infty \frac{r_{i}}{s_{i}} = r_1 + 
		\sum_{i=2}^\infty r_i(1+\frac{1}{s_i})\,.\label{rs}
	\end{align}
	Similarly, 
	\begin{align}
		\sum_{i=1}^\infty z_i x_i 
		&= \sum_{i=1}^\infty z_i(r_i - l_i) \nonumber\\
		&= \sum_{i=1}^\infty z_i r_i - \sum_{i=1}^\infty z_i 
		\frac{r_{i+1}}{s_{i+1}} 
		\nonumber\\
		&= z_1 r_1 + \sum_{i=2}^\infty r_i \left(z_i - 
		\frac{z_{i-1}}{s_i}\right)\,. 
		\label{zs}
	\end{align}
	By equating \eqref{rs} and \eqref{zs} term by term, we get the desired
	recurrence for $z_i$.
\end{proof}

The next lemma takes care of the case when online servers spawn too slowly.
\begin{lemma}\label{lemma:lower-slow}
	If the speeds in $\mathcal S$ satisfy 
	$\liminf_{n\to\infty}\sqrt[n]{\prod_{i=2}^ns_i}=1$ then 
	$\mathcal{S}$-DC is not competitive.
\end{lemma}
\begin{proof}
	For this lower bound we have requests on $n$ arbitrary positions in the 
	interval $[1,2]$, until $\mathcal S$-DC covers
	them all.
	
	The optimal cost is at most $2n$. This can be achieved by spawning a 
	fresh 
	server for each requested position.
	
	Since for every spawned online server we have
	$x_i \ge 1$, by Lemma \ref{lemma:online_cost} the online cost is
	$\cost=\sum_{i=1}^n z_i x_i \ge
	\sum_{i=1}^n z_i$. Unraveling the recurrence we get that $z_i = 1 +
	\frac{2}{s_i}  + \frac{2}{s_is_{i-1}} + \ldots + \frac{2}{s_i \cdot
		\ldots \cdot s_2}$. Thus,
	\begin{align}
		\cost &\ge n + \sum_{i=1}^{n-1} \sum_{j=1}^{n-i}  
		\frac{2}{\prod_{k=j+1}^{j+i}s_k}\nonumber\\
		&\ge \sum_{i=1}^{f(n)} \sum_{j=1}^{n-i}  
		\frac{2}{\prod_{k=j+1}^{j+i}s_k}\,.\label{eqn:slow_cost}
	\end{align}
	where
	\begin{align*}
		f(n) = \left\lfloor\frac{n}{2+2\log_2\prod_{i=2}^n s_i}\right\rfloor \le 
		\frac{n}{2}\,.
	\end{align*}
	We argue that for each $i=1,\dots,f(n)$, it holds for at least half of the 
	values of $j=1,\dots,n-i$ that $\prod_{k=j+1}^{j+i} s_k\le 2$. Indeed, 
	suppose this were not the case for some $i$. Let us partition the set 
	$J=\{1,\dots,n-i\}$ of $j$-values into subsets $J_0,\dots,J_{i-1}$, where 
	$J_m$ contains precisely those numbers from $J$ that are congruent to 
	$m$ modulo $i$. By assumption, we have $\prod_{k=j+1}^{j+i} s_k> 2$ 
	for at least half the values $j\in J$, so this must also be true for at least 
	half the values $j\in J_m$ for some $m$. However, this would mean that
	\begin{align*}
		\prod_{k=2}^n s_k \ge \prod_{j\in J_m}\prod_{k=j+1}^{j+i}s_k > 
		2^{|J_m|/2} \ge 2^{\floor{\frac{n-i}{i}}/2} \ge 2^{\frac{n}{2f(n)}-1} \ge 
		\prod_{i=2}^n s_i\,,
	\end{align*}
	a contradiction because the second inequality is strict.
	
	Thus, continuing from \eqref{eqn:slow_cost} we can further bound the 
	online cost as
	\begin{align*}
		\cost \ge f(n) \frac{n-f(n)}{2} \ge \frac{nf(n)}{4}\,.
	\end{align*}
	Since the optimal cost is at most $2n$, the competitive ratio is at least 
	$f(n)/8$. However, $f(n)$ is unbounded because
	\begin{align*}
		\frac{n}{2+2\log_2\prod_{i=2}^n s_i} = 
		\frac{1}{\frac{2}{n}+2\log_2\sqrt[n]{\prod_{i=2}^n s_i}}
	\end{align*}
	and the denominator in the last term gets arbitrarily close to $0$.
\end{proof}

The case of servers being spawned too aggressively is handled by the following lemma.
\begin{lemma}\label{lemma:lower-fast}
	If there exists an unbounded function $f(n)$ such that for each 
	$k\in\mathbb{N}$ we have $\prod_{i=k}^{k+n} s_{i} \ge f(n)$, then 
	$\mathcal{S}$-DC is not competitive. In particular, if 
	$\liminf_{i\to\infty}s_i>1$ then $\mathcal S$-DC is not competitive.
\end{lemma}
\begin{proof}
	Consider the following configuration of online positions and requests,
	denoted by circles and crosses respectively.
	
	
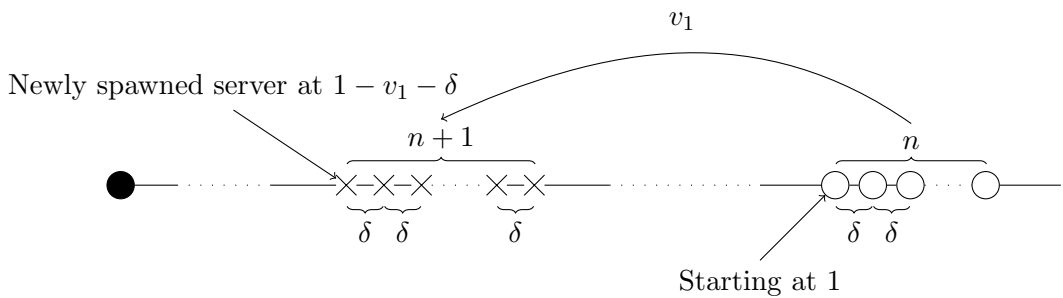
\begin{center}
		\begin{tikzpicture}
		\node[draw,shape=circle,fill] (source) at (-2.5,0) {};
		\node[draw,shape=circle] (p1) at (7,0) {};
		\node[draw,shape=circle] (p2) at (7.5,0) {};
		\node[draw,shape=circle] (p3) at (8.0,0) {};
		\node[draw,shape=circle] (p4) at (9.0,0) {};
		
		\node[draw,shape=cross out] (s1) at (0.5,0) {};
		\node[draw,shape=cross out] (s2) at (1,0) {};
		\node[draw,shape=cross out] (s3) at (1.5,0) {};
		\node[draw,shape=cross out] (s4) at (2.5,0) {};
		\node[draw,shape=cross out] (s5) at (3.0,0) {};
		
		\draw (source) -- (-1.75,0);
		\draw[loosely dotted] (-1.75,0) -- (-0.5,0);
		\draw (-0.5,0) -- (s1) -- (s2) -- (s3);
		\draw[loosely dotted] (s3) -- (s4);
		\draw (s4) -- (s5) -- (4,0);
		\draw[loosely dotted] (4,0) -- (6,0);
		\draw (6,0) -- (p1) -- (p2) -- (p3);
		\draw[loosely dotted] (p3) -- (p4);
		\draw (p4) -- (10,0);
		
		\draw[->,] (-1,1) node[above] {Newly spawned server at 
			$1-v_1-\delta$} -- (s1);
		\draw[->] (6,-1) node[below] {Starting at $1$} -- (p1);

		
		%
		\foreach \from/\to in {p1/p2, p2/p3, s1/s2, s2/s3, s4/s5}
		\draw[decoration={brace,mirror,raise=8pt},decorate]%
		($ (\from) + (0.01,0) $) -- node[below=10pt] (start) {$\delta$} ($ (\to) 
		+ %
		(-0.01,0) $);
		
		\draw[decoration={brace,raise=8pt},decorate]%
		($ (p1) + (0.01,0) $) -- node (start) [above=10pt] {$n$}  ($ (p4) + 
		(0.01,0) $);
		
		\draw[decoration={brace,raise=8pt},decorate]%
		($ (s1) + (0.01,0) $) -- node (finish) [above=10pt] {$n+1$}  ($ (s5) + 
		(0.01,0) $);
		
		\draw[->] ($ (start) + (0,0.25) $) to [out=145,in=25] node[above=5pt]%
		{$v_1$} ($ (finish) + (0,0.25) $); 
		
		\end{tikzpicture}
	\end{center}
	We start by spawning $n$ online servers grouped tightly, with
	the leftmost being at distance $1$ from the source and
	a very small gap $\delta$ between them. This is easily
	accomplished by repeating several requests on those
	points. Afterwards, we shift this group of $n$ servers (by means of
	requests on new $n+1$ points) to the left by
	$v_1$, chosen so that the $n+1$ points are covered exactly by the
	$n$ old servers plus a newly spawned one, which occupies the leftmost
	requested position $1-v_1-\delta$.
	
	This is repeated again and again, shifting each time the leftmost $n$ 
	spawned servers a new $v_k$ to the left via multiple requests on
	$n+1$ positions. The goal each time is to pull a new server from the
	source \textit{and} leave one behind forever, thus achieving an
	arbitrarily high competitive ratio for $\mathcal{S}$-DC variants that 
	spawn servers too fast.
	
	The offline cost can be calculated easily. The offline algorithm uses $n$
	servers to cover the first group of $n$ requested points in the interval 
	$[1,1+n\delta]$. Then it adds one more
	server and moves the group of $n+1$ servers to the left to satisfy all of 
	the following
	requests. At most, the group of offline servers will return close to
	the source, yielding an optimal cost of
	\begin{equation}\label{lower1:offline_cost}
	\OPT \le 2(n+1)(1+n\delta) = O(n)
	\end{equation}
	since $\delta$ is very small.
	
	To bound the online cost, we need to compute the values $v_k$ first. Let 
	$\ell_i^k$ and $r_i^k$ denote the cumulative distance to the left and 
	right respectively traveled by $x_i$ during the left shift by $v_k$ of the 
	group $x_k,x_{k+1},\dots,x_{k+n-1}$. The nonzero values among these 
	are
	\begin{align}
		\ell_k^k &= v_k \nonumber\\
		r_{k+1}^k &= v_ks_{k+1} \nonumber\\
		\ell_{k+1}^k &= v_k(1+s_{k+1})\nonumber\\
		r_{k+2}^k &= v_k(s_{k+2}+s_{k+1}s_{k+2})\nonumber\\
		\ell_{k+2}^k &= v_k(1+s_{k+2}+s_{k+1}s_{k+2})\nonumber\\
		\vdots \nonumber\\
		r_{k+n}^k &= v_k(s_{k+n} + s_{k+n-1}s_{k+n} + \ldots + 
		\prod_{j=k+1}^{k+n}s_j) =
		v_k \sum_{i=k+1}^{k+n} \prod_{j=i}^{k+n} s_j\,.\label{eq:rk+nk}
	\end{align}
	On the other hand, the new position of the server $x_{k+n}$ pulled from 
	the source during these moves is $1-\sum_{i=1}^{k}v_i-k\delta$. 
	Equating this with \eqref{eq:rk+nk} and solving for $v_k$ yields (and 
	assuming that $n$ is even)
	\begin{align*}
		v_k= \frac{1 - \sum_{i=1}^{k-1}v_i - k\delta}{1 +
			\sum_{i=k+1}^{k+n} \prod_{j=i}^{k+n} s_j} \le
		\frac{1}{\frac{n}{2} \prod_{j=k+\frac{n}{2}}^{k+n}s_j} \le 
		\frac{2}{nf(\frac{n}{2})}\,.
	\end{align*}
	
	We will calculate the number of repetitions before the left border of
	the group of servers (just) passes $\frac{1}{2}$.
	If $l$ is the number of repetitions, we have
	\begin{align*}
		\frac{1}{2} \le \sum_{k=1}^l (v_k + \delta) \le \frac{2l}{nf(\frac{n}{2})} 
		+ l\delta
	\end{align*}
	and for sufficiently small $\delta$ this means that
	\begin{align*}
		l \ge \frac{n}{5}f\left(\frac{n}{2}\right)
	\end{align*}
	
	If we do $l-1$ repetitions, then each of them will pull a new server at 
	least $1/2$ away from the source, resulting in an online cost of 
	$\Omega(n)\cdot f(\frac{n}{2})$. As the offline cost is $O(n)$ and $f(n)$ 
	is unbounded, the algorithm is not competitive.
\end{proof}

Since the sequence of speeds $s_i$ is monotonic and bounded from 
below by 1, we have either $\lim_{i \rightarrow \infty} s_i = 1$, in which case Lemma~\ref{lemma:lower-slow} applies, or otherwise $\liminf_{i \rightarrow \infty} s_i >1$ and Lemma~\ref{lemma:lower-fast} applies. In any case, the competitive ratio is 
unbounded:

\begin{theorem}\label{thm:dc-var}
	Algorithm $\mathcal{S}$-DC is not competitive for any $\mathcal{S}$.
\end{theorem}

\section{Reduction to Bounded Spaces}\label{sec:redBounded}

In this section we show a reduction from the infinite server problem on general metric spaces to bounded subspaces. Specifically, a metric space can be partitioned into ``rings'' of points whose distance from the source is between $r^n$ and $r^{n+1}$, where $r>1$ is fixed and $n\in\mathbb Z$. We show that if the infinite server problem is strictly $\rho$-competitive on each ring, then it is competitive on the entire metric space.

\begin{theorem}\label{thm:redBounded}
	Let $M$ be a metric space and $s\in M$ and let $r>1$. For $n\in\mathbb Z$ let $M_n=\{s\}\cup\{p\in M \mid d(s,p)\in[r^n,r^{n+1})\}$. If for each $n$ the infinite server problem on $(M_n,s)$ is strictly $\rho$-competitive, then on $(M,s)$ it is strictly $\frac{4r-1}{r-1}\rho$-competitive.
\end{theorem}
\begin{proof}
	Let $\ALG_n$ be a $\rho$-competitive algorithm for the infinite server problem on $(M_n,s)$.
	
	For a request sequence $\sigma$, let $\sigma_n$ be the subsequence of requests in $M_n$. Let $\ALG$ be the algorithm for $(M,s)$ that uses different servers for each of the subsequences $\sigma_n$ and serves them independently according to $\ALG_n$.
	
	The total online cost is $\ALG(\sigma)=\sum_{n} \ALG_n(\sigma_n)\leq \rho \sum_{n}\OPT(\sigma_n)$. To finish the proof, it suffices to show that
	\begin{align}
		\sum_{n}\OPT(\sigma_n) \le \frac{4r-1}{r-1}\OPT(\sigma)\,.\label{eq:redBoundedOpt}
	\end{align}
	Thus, we only need to analyze the offline cost. We do this for each offline server separately. Fix some offline server
	$x$. Let $N_0$ and $N_1$ be the minimal and maximal values of $n$ such that $x$ visits $M_n$. We can assume without loss of generality (by adding virtual points to the metric space) that whenever $x$ moves from $M_n$ to $M_{n'}$ for some $n<n'$, it travels across points $p_{n+1},p_{n+2},\dots,p_{n'}$ with $d(s,p_i)=r^{i}$, and similarly for $n>n'$.
	
	The movements of server $x$ can be tracked by many servers, one
	server $x_n$ in every set $M_n$ for $N_0\leq n\leq
	N_1$. When server $x$ is in $M_n$, server $x_n$ is
	exactly at the same position tracking the movement of $x$. When server $x$ exits $M_n$ at some point $p$ at the boundary to $M_{n-1}$ or $M_{n+1}$, server $x_n$ freezes at $p$.
	The movement cost of $x_n$ can be partitioned into the cost of deploying $x_n$ at the first point visited in $M_n$, the tracking cost within $M_n$, and the cost of of relocating $x_n$ whenever $x$ re-enters $M_n$ at a location different from the last exiting location.
	
	The total tracking cost of all servers $x_n$ is bounded by the distance traveled by $x$. The cost of deploying all
	servers $x_n$ is $\sum_{n=N_0}^{N_1}r^{n}\le\sum_{n=-\infty}^{N_1}r^{n}=r^{N_1+1}/(r-1)$, which is at
	most $\frac{r}{r-1}$ times the total movement of server $x$, because the
	latter is at least $r^{N_1}$.
	
	To bound the relocating cost, say $x$ exits $M_n$ at $p$ and re-enters it at $p'$. Then $p$ and $p'$ are at the boundary of $M_n$ and $M_{n+u}$ for $u\in\{-1,+1\}$. Let $b$ be the distance traveled by $x$ in $M_{n+u}$ between the times when it is entered at $p$ and when it is next exited. If this exiting is at $p'$, then the relocating cost $d(p,p')$ is at most $b$ by the triangle inequality. Otherwise, $x$ exits $M_{n+u}$ at a point $p''$ at the boundary of $M_{n+u}$ and $M_{n+2u}$. If $u=1$, then $d(p,p')\le d(s,p)+d(s,p')= 2r^{n+1}$ and $b\ge d(p,p'')\ge d(s,p'') - d(s,p) = r^{n+2} - r^{n+1}=(r-1)r^{n+1}$. If $u=-1$, then $d(p,p')\le d(s,p)+d(s,p')= 2r^{n}$ and $b\ge d(p,p'')\ge d(s,p) - d(s,p'') = r^{n} - r^{n-1}=\frac{r-1}{r}r^{n}$. In both cases, the relocating cost $d(p,p')$ is at most $\frac{2r}{r-1}b$. Thus, the total relocating cost of all servers $x_n$ is at most $\frac{2r}{r-1}$ times the total distance traveled by $x$.
	
	Thus, the sum of deployment, tracking and relocating cost of the servers $x_n$ is at most $\frac{4r-1}{r-1}$ times the distance traveled by $x$. This shows \eqref{eq:redBoundedOpt}, giving the
	statement of the theorem.
\end{proof}

The last theorem can also be slightly generalized to the case where instead of \emph{strict} $\rho$-competitiveness, an additive term proportional to $r^n$ is allowed. It is not difficult to show the following specialization for the line, where the premise can be weakened to require competitiveness only on a single interval:

\begin{corollary}
	Let $0<a<b$. The infinite server problem is competitive on the line if and only if it is competitive on $(\{0\}\cup[a,b],0)$.
\end{corollary}

Another consequence of Theorem~\ref{thm:redBounded} is a reduction to spaces where the source is at a uniform distance from all other points.

\begin{corollary}
	Suppose there exists $\rho$ so that the infinite server problem is strictly $\rho$-competitive on any metric space where the distance from the source to any other point is the same. Then the infinite server problem on general metric spaces is competitive.
\end{corollary}
\begin{proof}
	Follows from Theorem~\ref{thm:redBounded} by increasing the distance from $s$ to the other points in $M_n$ to $r^{n+1}$, making a multiplicative error of at most $r$.
\end{proof}

\section{Open Problems}\label{sec:concl}
The most obvious open problem is whether the infinite server problem is competitive on general metric spaces. A challenging special case is to resolve the question for the real line. Similarly, improving the MOO 
algorithm and settling the question for layered graphs remains open. It would also be interesting to find a metric space with a competitive ratio greater than $3.146$. 
Another possible line of research is to consider randomized algorithms.

\bibliography{bibliography}{}
\bibliographystyle{plainurl}

\end{document}